\documentclass[12pt,letterpaper]{amsart}

 \usepackage{txfonts} 
\usepackage{graphicx}
\usepackage{verbatim}
\usepackage{tikz}

\usepackage{caption}
\usepackage{subcaption}

\usepackage[colorlinks,linkcolor=blue,citecolor=blue]{hyperref}

 \theoremstyle{plain}

\newtheorem{theorem}{Theorem}[section]

\theoremstyle{definition}

\newtheorem{remark}[theorem]{Remark}

\theoremstyle{remark}
\numberwithin{equation}{section}

\makeatletter
\@namedef{subjclassname@1991}{2020 Mathematics Subject Classification}
\makeatother
\begin{document}

\captionsetup[figure]{labelfont={bf},labelformat={default},labelsep=period,name={Fig.}}

\title[A gradient descent algorithm for computing circle patterns]{A gradient descent algorithm for computing circle patterns}

\author{Te Ba}
\address{School of Mathematics, Hunan University, Changsha 410082, China}
\email{batexu@hnu.edu.cn}
\author{Ze Zhou}
\address{School of Mathematical Sciences, Shenzhen University,  Shenzhen 518060, China }
\email{zhouze@szu.edu.cn}

\date{}


\begin{abstract}
This paper presents a new algorithm for generating planar circle patterns. The algorithm employs gradient descent and conjugate gradient method to compute circle radii and centers separately. Compared with existing algorithms, the proposed method is more efficient in computing centers of circles and is applicable for realizing circle patterns with possible obtuse overlap angles.
\end{abstract}

\maketitle

\section{Background}

Circle patterns are configurations of circles with prescribed combinatorial structures.
They were first studied by Koebe~\cite{zbMATH03027510} and received widespread attention when rediscovered by Thurston~\cite{MR1435975} in his study of $3$-manifolds.
Thurston~\cite{thurston1985finite} further conjectured that circle patterns could be used to approximate conformal mappings. This was later confirmed by Rodin and Sullivan~\cite{MR906396}, paving the way for a wide range of practical applications of circle patterns in computational conformal geometry~\cite{10.1145/1138450.1138461}, medical imaging~\cite{HURDAL2004S119}, mesh generation~\cite{MR1791192} and others.
Many of these applications rely on constructing circle patterns on geometric structures, which has spurred the development of various computational approaches.
Representative algorithms include those based on variational principle~\cite{colin1989empilements}, Thurston's continuous method~\cite{MR1975216}, Beurling's boundary value problem approach~\cite{MR2886749}, and radius-center alternating iterations~\cite{MR3638944}.
For a comprehensive survey, we refer to the works of Stephenson~\cite{MR2131318} and Bowers~\cite{MR4264578}.

The aim of this paper is to develop a new algorithm for realizing planar circle patterns, offering a practical framework for constructing conformal mappings from a simply connected domain onto a polygon with prescribed interior angles.

\section{Computational Methods for Circle Patterns}
\subsection{Preparations}
Let $D$ be a topological polygon.
Let $T=(V,E,F)$ be a triangulation of $D$, and $\mu$ a metric on $D$ such that $(D,\mu)$ is a Euclidean simple polygon.
A circle pattern $\mathcal P$ on $(D,\mu)$ is a collection of circles centered in $D$ and $\mathcal P$ is said to be $T$-type if there exists a geodesic triangulation $T_\mu$ of $(D,\mu)$ satisfies the following properties:
($i$) $T_\mu$ is isotopic with $T$;
($ii$) The vertices of $T_\mu$ coincide with the center of $\mathcal P$.
It has been shown by Stephenson~\cite{MR2131318} and generalized by Ge, Hua and Zhou~\cite{MR4334399} and by Jiang, Luo and Zhou~\cite{MR4109913} that for any triangulation $T$ of $D$, there exists a unique (up to similarities) metric $\mu$ on $D$ such that $(D,\mu)$ has prescribed interior angles and supports a $T$-type circle pattern $\mathcal{P}$ with prescribed overlap angles between each pair of intersecting circles.

To obtain such circle patterns, we first recall Thurston's construction~\cite[Chap. 13]{MR1435975}.
Let $\Theta:E\to[0,\pi)$ be an overlap angle weight and let $r\in\mathbb{R}_+^{|V|}$ be a vector assigning each $v\in V$ a radius of $r_v$. For each edge $uv\in E$, the edge length $l_{uv}$ is given by
\begin{equation}\label{E-1}l_{uv}=\sqrt{r_u^2+r_v^2+2r_ur_v\cos\Theta_{uv}}.\end{equation}
For each face $uvw\in F$, let us denote
\begin{equation}\label{E-2}I_u^{vw}=\cos \Theta_{vw} + \cos \Theta_{uv} \cos \Theta_{uw}.\end{equation}
Then we can construct a Euclidean triangle $\Delta_{uvw}$ with side lengths $l_{uv}$, $l_{vw}$ and $l_{uw}$ under the condition $I_u^{vw}\geq0$, $I_v^{uw}\geq0$ and $I_w^{uv}\geq0$ (see~\cite{MR4334399,MR4109913}).
By gluing these triangles along their common edges, one obtains a cone metric on $D$ with singularities at the vertices in $V$.
For each vertex $v\in V$, let $\sigma_v$ denote the sum of the interior angles of all triangles incident to $v$ and let \( V_{\partial} \subset V \) denote the set of boundary vertices.
The discrete curvature at each vertex is defined by
\[
K_v =
\begin{cases}
\sigma_v-2\pi, & \text{if } v \in V\setminus V_{\partial},\\
\sigma_v-\theta_v, & \text{if } v\in V_{\partial},
\end{cases}
\]
where $\theta_v$ is the prescribed interior angle at $v\in V_{\partial}$.
We seek a radius vector $r_{\mathcal P}$ such that the cone singularities vanish, i.e., $K_v(r_{\mathcal P})=0$ for each $v\in V$.
With this radius vector, we can then construct the desired circle pattern realizing the prescribed data (see~\cite[Chap. 13]{MR1435975}).

\subsection{Algorithms}
A $T$-type circle pattern $\mathcal P$ can be uniquely realized by its radii and centers, denoted as $\mathcal P=(r,z)\in\mathbb{R}^{|V|}_+\times\mathbb{C}^{|V|}$.
In this subsection, we present an algorithm for computing the radii and centers of $\mathcal P$ separately.

The first step is to compute the target radii.
We introduce an energy function
 $\mathcal E:\mathbb{R}^{|V|}_+\to\mathbb{R}$ defined by
\[\mathcal E(r)=\sum_{v\in V} K_v(r)^2.\]
Evidently, a radius vector of $\mathcal P$ is a critical point of $\mathcal E$, as it attains its minimum.
Moreover, Ge, Hua and Zhou~\cite{MR4334399}, Li, Luo and Xu~\cite{MR4781993} proved that \(\mathcal{E}\) admits no other critical points.
Hence, gradient descent can be applied to $\mathcal E$ to approximate the target radius vector.
Given a test radius vector $r^{(0)}\in\mathbb{R}_+^{|V|}$ and a step size \(\eta > 0\), the sequence of approximate radii \(\{ r^{(k)} \}_{k\in\mathbb{N}}\) is generated by
\[
 r^{(k+1)}  =  r^{(k)}  - \eta\, \nabla \mathcal E( r^{(k)} ).
\]

\begin{remark}
The approach of computing circle pattern radii by minimizing an energy function was introduced by Colin de Verdi{\`e}re~\cite{colin1989empilements} and developed by Chow and Luo~\cite{MR2015261}.
As an improvement, the energy $\mathcal E$ can be optimized via nonlinear least-squares methods to speed up convergence.
\end{remark}

The second step is to determine the circle centers of $\mathcal P$.
We first fix the centers of a pair of adjacent boundary circles.
Then the center assignment $z_{\mathcal P}\in\mathbb C^{|V|}$ of $\mathcal P$ is uniquely determined.
Recall that $\theta_v$ is the prescribed interior angle at $v$. Using the radius vector obtained in the first step, the centers of the remaining boundary circles can be successively computed, since each triple of consecutive boundary vertices $u$, $v$ and $w$ satisfies
\begin{equation}\label{E-5}
|z_{v}-z_{w}|=l_{vw},\quad\angle z_uz_vz_w=\theta_v.
\end{equation}
Next, we define the discrete Dirichlet energy $\Psi:\mathbb{C}^{|V|}\to\mathbb{R}$ by
\begin{equation}\label{E-3}\Psi(z)=\sum_{uv\in E}\frac{\partial K_u}{\partial \log r_v}(r_{\mathcal P})|z_u-z_v|^2.\end{equation}
Owing to the results of Ge Hua and Zhou~\cite{MR4334399}, the coefficient can be written as
\[
\begin{aligned}
\frac{\partial K_u}{\partial \log r_v}=\frac{\partial K_v}{\partial \log r_u}&=\sum_{uvw\in F}\frac{r_u r_v\left[\sin^2 \Theta_{uv}r_u r_v + (I_u^{vw}r_u + I_v^{uw}r_v)r_w\right]}{l^2_{uv}A_{uvw}},\\
\end{aligned}
\]
where $l_{uv}$ is given in equation~\eqref{E-1}, $I_u^{vw}$ is given in equation~\eqref{E-2} and $A_{uvw}$ is the area of the triangle with side lengths $l_{uv}$, $l_{vw}$, $l_{uw}$.
 Dubejko~\cite{MR1476986} shows that 
\begin{equation}\label{E-4}\sum_{u\sim v}\frac{\partial K_u}{\partial \log r_v}(r_{\mathcal P})(z_{u}-z_{v})\Big|_{z=z_{\mathcal P}}=0.\end{equation}
Here $u\sim v$ means $u$ is a neighbor of $v$.
By substituting the center vector of boundary circles into \(\Psi\), combining~\eqref{E-3} and~\eqref{E-4}, we see that the center vector of the interior circles is a critical point of \(\Psi\).
Note that $\frac{\partial K_u}{\partial \log r_v}(r_{\mathcal P})\geq0$ under the condition  $I_u^{vw}\geq0$ for each $uvw\in F$.
Hence, \(\Psi\) is a quadratic function of the interior circle centers, with a symmetric positive definite Hessian.
For numerical computation, the conjugate gradient method provides an efficient way to minimize $\Psi$.
Let $x^{(0)}=(1,\cdots,1)^{\mathrm{T}}$ and $y^{(0)}=(1,\cdots,1)^{\mathrm{T}}$.
The iteration can be performed separately for the $x$- and $y$-coordinates by
\[
x^{(k+1)} = x^{(k)} + l^{(k)}\alpha^{(k)}, \quad
y^{(k+1)} = y^{(k)} + d^{(k)} \beta^{(k)},
\]
where $l^{(k)}$, $d^{(k)}$ (resp. $\alpha^{(k)}$, $\beta^{(k)}$) denote the step sizes (resp. descent directions), obtained by applying the conjugate gradient method applied to $\Psi$.
Combining the previously obtained approximation of the boundary circle centers, we obtain an approximation $z^{(n)}\in\mathbb{C}^{|V|}$ to the center vector of circles of $\mathcal P$.
\begin{remark}
From the perspective of geometric analysis on graphs, Equation~\eqref{E-4} shows that the centers of $\mathcal P$ are discrete harmonic on the $1$-skeleton of $T$ and $\Psi$ is the associated discrete Dirichlet energy.
This discrete harmonicity also applies to computing circle centers for planar circle patterns with prescribed radii of boundary circles as discussed in Ba, Li and Xu~\cite{MR4531772}, Wegert, Roth and Kraus~\cite{MR2886749}.
\end{remark}

An outline of the numerical procedure for computing circle patterns is summarized in Table~\ref{TA-1} and the convergence of the algorithm is given in Theorem~\ref{T-1}.

\begin{table}[h]
\caption{Algorithm Pseudocode}
    \centering
\begin{tabular}{lp{0.85\linewidth}}
\hline
\textbf{Input:}  & $T=(V,E,F)$, overlap angle weight $\Theta:E\to[0,\pi)$, interior angle weight $\theta:V_{\partial}\to(0,\pi]$. \\
\hline
\multicolumn{2}{c}{\textbf{Step 1: Computing Circle Radii}} \\
\hline
\textbf{Initialize:} &   Initial radius vector $r^{(0)}$, step size $\eta$, tolerance $\varepsilon$. \\
\textbf{Repeat:}

& 1. Compute $\nabla\mathcal E(r_k)$.\\
  &2. Update $r^{(k+1)} =r^{(k)} - \eta \nabla\mathcal E(r^{(k)})$.\\
  &3.  $k \leftarrow k + 1$.\\
\textbf{Until:} & $|\nabla\mathcal E(r^{(k)})| < \varepsilon$.\\

\textbf{Output:} & Approximate radius vector $r^{(k)}$ \\
\hline

   \multicolumn{2}{c}{\textbf{Step 2: Computing Circle Centers}} \\
   \hline

\textbf{Initialize:} & $r^{(k)}$, boundary centers $z_1,\cdots ,z_m$, initial center vector $z^{(0)}$, tolerance $\varepsilon$.\\
\textbf{Repeat:}
& 1.~Compute  $ l^{(n)}$, $ d^{(n)}$, $\alpha^{(n)}$, $\beta^{(n)}$.\\
& 2.~Update $x^{(n+1)} = x^{(n)} + l^{(n)}\alpha^{(n)}$, $y^{(n+1)} = y^{(n)} + d^{(n)}\beta^{(n)}$.\\
& 3.~Update $z^{(n+1)} =(x^{(n+1)},y^{(n+1)})$.\\
& 4.~$n \leftarrow n + 1$.\\
\textbf{Until:} &$\|\nabla\Psi(z^{(n)})\| < \varepsilon$.\\

\textbf{Output:} &  Approximate center vector  $z^{(n)}$. \\
\hline

\end{tabular}
\label{TA-1}
\end{table}

\begin{theorem}\label{T-1}
For any $\varepsilon>0$, there exists a step size $\eta\in\mathbb{R}_+$ and $k,n\in\mathbb{N}$ such that
\[
\| r^{(k)} - r_{\mathcal P} \| < \varepsilon, \quad
\| z^{(n)} - z_{\mathcal P} \| < \varepsilon.
\]

\end{theorem}
\begin{proof}
Given a radius vector $r\in\mathbb{R}_{+}^{|V|}$, let $z(r)\in\mathbb{C}^{|V|}$ denote the exact center vector of $\mathcal P$ determined by~\eqref{E-5} and~\eqref{E-4}.
Obviously, $z(r)$ is $C^1$ with respect to $r$.
Then there exists $\delta, L>0$ such that
\[\Vert z(r)-z_{\mathcal P}\Vert\leq L\Vert r-r_{\mathcal P}\Vert\]
when $\Vert r-r_{\mathcal P}\Vert<\delta$.
Due to the convergence of gradient descent, for any $\varepsilon>0$, there exists a step size $\eta>0$ and $k_0\in\mathbb{R}_+$, $n_0\in\mathbb{N}$ such that
\[\Vert r^{(k_0)}-r_{\mathcal P}\Vert<\min\{\varepsilon,\frac{\varepsilon}{2L},\delta\},\quad\Vert z^{(n_0)}-z(r^{(k_0)})\Vert<\frac{\varepsilon}{2}.\]
Then we obtain
\[\Vert z^{(n_0)} - z_{\mathcal P}\Vert \leq\Vert z^{(n_0)}-z(r^{(k_0)})\Vert + \Vert z(r^{(k_0)})-z_{\mathcal P}\Vert<\frac{\varepsilon}{2}+L\frac{\varepsilon}{2L}=\varepsilon,\]
which completes the proof.
\end{proof}

\subsection{Numerical Experiments}
The existence of circle patterns with possibly obtuse overlap angles has been studied in a recent series of works~\cite{MR4109913,zhou2019,zhou2021,Zhou2025}.
Our algorithm is compatible with this generalized framework. It enriches the theory of computational realization for circle patterns developed in~\cite{MR1975216,MR3638944,colin1989empilements,MR2886749}, which mainly focuses on tangent circle patterns.
See Fig.~\ref{F-2} for some illustrations generated by our algorithm.
As an application, this provides a practical method for constructing conformal maps from a simply connected domain onto a convex polygon with prescribed interior angles, following the discrete conformal approach pioneered by Thurston~\cite{thurston1985finite} and Rodin and Sullivan~\cite{MR906396}.

\begin{figure}[h]

    \begin{minipage}{0.46\textwidth}
    \centering
        \includegraphics[width=0.97\textwidth]{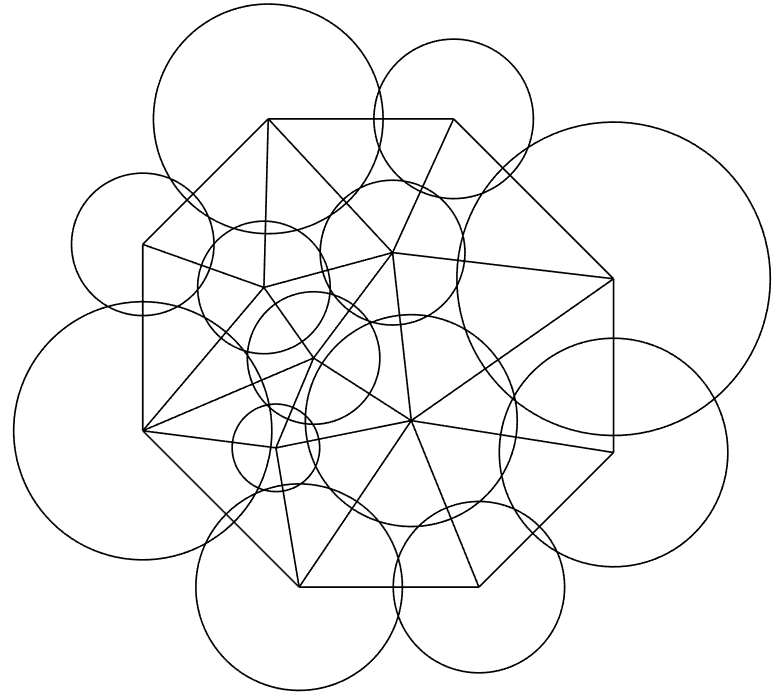}
    \end{minipage}
    \hfill
    \begin{minipage}{0.46\textwidth}
        \centering
        \includegraphics[width=0.86\textwidth]{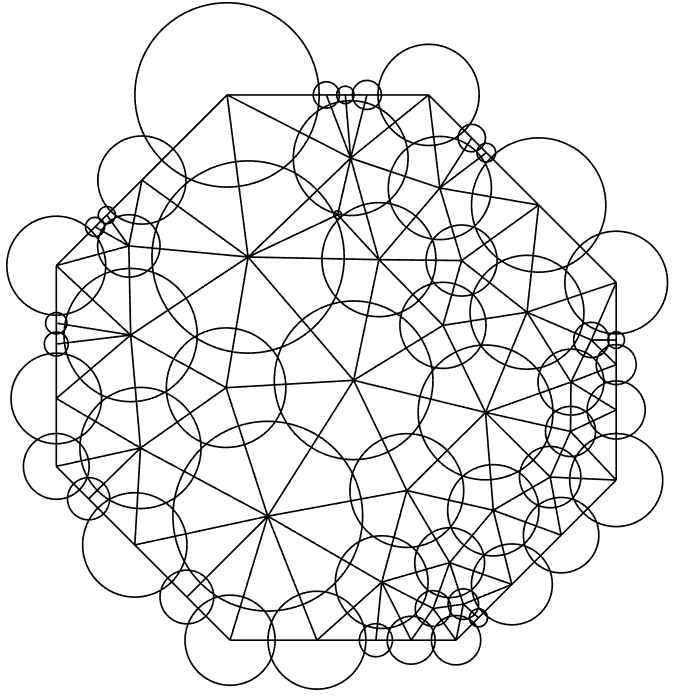}
    \end{minipage}
    \caption{Generating Circle Patterns with Prescribed Data}
    \label{F-2}
\end{figure}

\section{Further Discussion}
In conclusion, we pose the question of whether the proposed algorithm can be extended to realize hyperbolic circle patterns.
The main challenge is to exploit the discrete harmonic property of circle centers described in equation~\eqref{E-2}.
Alternatively, the circle centers may be sequentially determined in accordance with the prescribed combinatorial structure.
However, efficient methods for computing the centers remain a challenge for large-scale circle patterns.

\section*{Acknowledgments}
Te Ba and Ze Zhou were supported by NSFC Grant No. 12371075 and Grant No. 12422105.


%
%
%
%


\begin{thebibliography}{10}
\providecommand{\url}[1]{\texttt{#1}}
\providecommand{\urlprefix}{URL }
\providecommand{\doi}[1]{https://doi.org/#1}

\bibitem{MR4531772}
Ba, T., Li, S., Xu, Y.: Rigidity of bordered polyhedral surfaces. Calc. Var.
  Partial Differential Equations  \textbf{62}(3),  Paper No. 78, 20 (2023)

\bibitem{MR1791192}
Bern, M., Eppstein, D.: Quadrilateral meshing by circle packing. Internat. J.
  Comput. Geom. Appl.  \textbf{10}(4),  347--360 (2000)

\bibitem{MR4264578}
Bowers, P.L.: Combinatorics encoding geometry: the legacy of {B}ill {T}hurston
  in the story of one theorem. In: In the tradition of {T}hurston---geometry
  and topology, pp. 173--239. Springer, Cham (2020)

\bibitem{MR2015261}
Chow, B., Luo, F.: Combinatorial {R}icci flows on surfaces. J. Differential
  Geom.  \textbf{63}(1),  97--129 (2003)

\bibitem{MR1975216}
Collins, C.R., Stephenson, K.: A circle packing algorithm. Comput. Geom.
  \textbf{25}(3),  233--256 (2003)

\bibitem{MR1476986}
Dubejko, T.: Random walks on circle packings. In: Lipa's legacy ({N}ew {Y}ork,
  1995), Contemp. Math., vol.~211, pp. 169--182. Amer. Math. Soc., Providence,
  RI (1997)

\bibitem{MR4334399}
Ge, H., Hua, B., Zhou, Z.: Circle patterns on surfaces of finite topological
  type. Amer. J. Math.  \textbf{143}(5),  1397--1430 (2021)

\bibitem{HURDAL2004S119}
Hurdal, M.K., Stephenson, K.: Cortical cartography using the discrete conformal
  approach of circle packings. NeuroImage  \textbf{23},  119--128 (2004)

\bibitem{MR4109913}
Jiang, Y.P., Luo, Q., Zhou, Z.: Circle patterns on surfaces of finite
  topological type revisited. Pacific J. Math.  \textbf{306}(1),  203--220
  (2020)

\bibitem{10.1145/1138450.1138461}
Kharevych, L., Springborn, B., Schr\"{o}der, P.: Discrete conformal mappings
  via circle patterns  \textbf{25}(2) (2006)

\bibitem{zbMATH03027510}
Koebe, P.: Kontaktprobleme der konformen {Abbildung}. Ber. {Verh}. {S{\"a}chs}.
  {Akad}. {Leipzig}  \textbf{88}(3),  141--164 (1936)

\bibitem{MR4781993}
Li, S., Luo, Q., Xu, Y.: Combinatorial {C}alabi flow on surfaces of finite
  topological type. Proc. Amer. Math. Soc.  \textbf{152}(9),  4035--4047 (2024)

\bibitem{MR3638944}
Orick, G.L., Stephenson, K., Collins, C.: A linearized circle packing
  algorithm. Comput. Geom.  \textbf{64},  13--29 (2017)

\bibitem{MR906396}
Rodin, B., Sullivan, D.: The convergence of circle packings to the {R}iemann
  mapping. J. Differential Geom.  \textbf{26}(2),  349--360 (1987)

\bibitem{MR2131318}
Stephenson, K.: Introduction to circle packing: The theory of discrete analytic
  functions. Cambridge University Press, Cambridge (2005)

\bibitem{MR1435975}
Thurston, W.P.: Three-dimensional geometry and topology., Princeton
  Mathematical Series, vol.~35. Princeton University Press, Princeton, NJ
  (1997)

\bibitem{thurston1985finite}
Thurston, W.P.: The finite riemann mapping theorem. In: An International
  Symposium at Purdue University in Celebration of de Branges{'} Proof of the
  Bieberbach Conjecture, Invited talk (1985)

\bibitem{colin1989empilements}
Colin~de Verdi{\`e}re, Y.: Empilements de cercles: convergence d'une methode de
  point fixe. S{\'e}minaire de th{\'e}orie spectrale et g{\'e}om{\'e}trie
  \textbf{6},  23--31 (1989)

\bibitem{MR2886749}
Wegert, E., Roth, O., Kraus, D.: On {B}eurling's boundary value problem in
  circle packing. Complex Var. Elliptic Equ.  \textbf{57}(2-4),  397--410
  (2012)

\bibitem{zhou2019}
Zhou, Z.: Circle patterns with obtuse exterior intersection angles. Preprint,
  arXiv:1703.01768  (2019)

\bibitem{zhou2021}
Zhou, Z.: Producing circle patterns via configurations. Preprint,
  arXiv:2010.13076  (2021)

\bibitem{Zhou2025}
Zhou, Z.: Generalizing {{Andreev's Theorem}} via circle patterns. Mathematische
  Annalen  (2025). \doi{10.1007/s00208-025-03286-4}

\end{thebibliography}

\end{document}